\documentclass[12pt]{elsarticle}
\usepackage[margin=2.5cm]{geometry}
\usepackage{lipsum}
\usepackage{tikz,pgfplots}
\usepackage{verbatim,pslatex}
\usepackage{booktabs,colortbl}
\usepackage{textcomp}
\usepackage{colortbl}
\definecolor{mygray}{gray}{.9}
\definecolor{mypink}{rgb}{.99,.91,.95}
\definecolor{mycyan}{cmyk}{.3,0,0,0}

\makeatletter
\def\ps@pprintTitle{%
   \let\@oddhead\@empty
   \let\@evenhead\@empty
   \def\@oddfoot{\reset@font\hfil\thepage\hfil}
   \let\@evenfoot\@oddfoot
}
\makeatother

\usepackage{amsfonts,color,morefloats,pslatex}
\usepackage{amssymb,amsthm, amsmath,latexsym}

\newtheorem{theorem}{Theorem}
\newtheorem{lemma}[theorem]{Lemma}
\newtheorem{proposition}[theorem]{Proposition}
\newtheorem{corollary}[theorem]{Corollary}

\newtheorem{open}[theorem]{Open Problem}

\newtheorem{example}[theorem]{Example}

\newtheorem{conj}[theorem]{Conjecture}

\newcommand{\tr}{{\mathrm{Tr}}}

\newcommand{\gf}{{\mathrm{GF}}}
\newcommand{\C}{{\mathcal{C}}}

\newcommand{\wt}{{\mathrm{wt}}}
\newcommand{\dist}{{\mathrm{dist}}}
\newcommand{\supp}{{\mathrm{Supp}}}
\newcommand{\bc}{{\mathbf{c}}}

\newcommand{\srm}{{\mathrm{SRM}}}
\newcommand{\RM}{{\mathrm{RM}}}
\newcommand{\prm}{{\mathrm{PRM}}}

\newcommand{\ai}{{\mathrm{AI}}}
\usepackage{blindtext}

\begin{document}
\begin{frontmatter}



\title{A Novel Application of Boolean Functions with High Algebraic Immunity in Minimal Codes
\tnotetext[fn1]{C. Ding's research was supported by the Hong Kong Research Grants Council,
Proj. No. 16300418. C. Tang was supported by National Natural Science Foundation of China (Grant No.
11871058) and China West Normal University (14E013, CXTD2014-4 and the Meritocracy Research
Funds).}
}

\author[ch]{Hang Chen}
\ead{chenhangxihua@163.com}

\author[csd]{Cunsheng Ding}
\ead{cding@ust.hk}

\author[sm]{Sihem Mesnager}
\ead{smesnager@univ-paris8.fr}

\author[ch]{Chunming Tang}
\ead{tangchunmingmath@163.com}
\address[ch]{School of Mathematics and Information, China West Normal University, Nanchong, Sichuan,  637002, China}
\address[csd]{Department of Computer Science and Engineering, The Hong Kong University of Science and Technology, Clear Water Bay, Kowloon, Hong Kong, China}
\address[sm]{LAGA, Department of Mathematics, Universities of Paris VIII and Paris XIII, CNRS, UMR 7539 and Telecom ParisTech, France}





\begin{abstract}
Boolean functions with high algebraic immunity
are important cryptographic primitives in some stream ciphers.
In this paper, two methodologies for constructing binary minimal codes from
sets, Boolean functions and vectorial Boolean functions with high algebraic immunity are proposed.
More precisely, a general construction of new minimal codes using minimal codes contained
in Reed-Muller codes and sets without nonzero low degree annihilators is presented.
The other construction allows us to yield minimal codes from
certain subcodes of Reed-Muller codes and vectorial Boolean functions
with high algebraic immunity.
 Via these general constructions, infinite families of minimal binary linear codes of dimension $m$
 and length less than or equal to $m(m+1)/2$ are obtained.
 In addition, a lower bound on the minimum distance of the proposed minimal linear codes is established.
 Conjectures and open problems are also presented.  The results of this paper show that
 Boolean functions with high algebraic immunity have nice applications in several fields such as symmetric cryptography,
  coding theory and secret  sharing schemes.
 \end{abstract}


\end{frontmatter}


\section{Introduction}

Secret sharing, independently introduced in 1979 by Shamir \cite{Shamir79} and Blakley \cite{Blakley79}, 
is one of the most widely studied topics in cryptography. Relations between linear codes and secret sharing schemes 
were first investigated by McEliece and Sarwate in \cite{McESar81}. In theory every linear code can be employed to
construct secret sharing schemes. Unfortunately, it is extremely hard to determine the access structures of secret 
sharing schemes based on general linear codes. However,  the access structures of secret sharing schemes based 
on minimal linear codes are known and interesting \cite{DY03,Mas93,Mas95}.

Minimal codes have already received a lot of attention. It was pointed out in \cite{ABN19} and \cite{TQLZ19}
that minimal codes are close to blocking sets
in finite geometry.
Many minimal linear codes
were obtained from codes with few weights \cite{DD15,DY03,M17,MOS17,MOS18,TLQZT}.
Recently, Ding, Heng and Zhou \cite{DHZI} constructed
three infinite families of minimal binary linear codes using certain Boolean functions.
They also constructed an infinite family of minimal ternary linear codes from ternary functions in \cite{DHZF}.
Bartoli and Bonini \cite{BBI} generalized the construction of minimal linear codes in \cite{DHZF} from the ternary  case to the 
odd $p$ characteristic case via $p$-ary functions. Li and Yue \cite{LY} obtained some minimal binary linear codes with nonlinear Boolean functions.
Xu and Qu \cite{XQ19} constructed
minimal $q$-ary linear codes from some special functions.
In the recent paper \cite{MQRT},  the authors considered minimal codes from the supports of $p$-ary functions.
 Lu, Wu and Cao \cite{LWC} obtained minimal codes with
special subsets of vector spaces over finite fields.
Bonini and Borello \cite{BB19} presented a family of minimal codes arising from some blocking sets.

The main objective of this paper is to find connections among special sets, Boolean functions with high algebraic immunity and binary minimal  codes. Two general constructions of minimal binary codes with minimal codes contained in the Reed-Muller codes,
subsets of finite fields and vectorial Boolean functions with high algebraic immunity
are proposed.  Two families of minimal codes contained in the second-order Reed-Muller code
with large dimension are presented.  Sets and vectorial Boolean functions with high algebraic immunity
are also demonstrated. By plugging  these subcodes, special sets and vectorial Boolean functions into our general construction, 
some  infinite classes of minimal binary linear codes of dimension $m$ and length less than or equal to $m(m+1)/2$ are produced.
Finally, a lower bound on the minimum distance of the proposed minimal codes is derived. 
Conjectures and open problems are also presented.

The rest of this paper is organized as follows. In Section \ref{sec:intr}, we recall some standard facts about 
cyclic codes, Reed-Muller codes and vectorial Boolean functions. In Section  \ref{sec:codes-sets},
we establish some relations between binary minimal codes and subsets of finite fields without nonzero low degree annihilators.
It enables us to yield minimal codes via certain subcodes of Reed-Muller codes and sets with high algebraic immunity.
In Section \ref{sec:codes-vec-func}, we present a general construction of minimal codes from subcodes of Reed-Muller codes and
vectorial Boolean functions having high algebraic immunity.
In Section \ref{sec:conc}, we conclude this paper.

\section{Background}\label{sec:intr}

\subsection{Boolean functions and vectorial Boolean functions}
A \emph{Boolean function} $f$ on $\gf(2^m)$ is a $\gf(2)$-valued function on the Galois field $\gf(2^m)$ of order $2^m$.
The set of all Boolean functions over $\gf(2^m)$ forms a ring and is denoted by $\mathbb{B}_m$.
The \emph{support} of $f$, denoted by $\supp(f)$, is the set of elements of $\gf(2^m)$ whose image under $f$ is $1$, that is,
$\supp(f)=\left \{  x\in \gf(2^m): f(x)=1 \right \}$.
The \emph{Hamming weight $\wt (f)$} of a Boolean function is the size of its support $\supp (f)$.
The \emph{characteristic function} $f_D$ of a subset $D$ of $\gf(2^m)$
is the Boolean function such that $f(x)=1$ for all  $x \in D$ and $f(x)=0$ for all  $x \in \gf(2^m) \setminus D$.
Thus $\supp\left (f_D \right )=D$.
Every nonzero Boolean function $f$ on $\gf(2^m)$ has a unique univariate polynomial expansion of the form
\begin{eqnarray*}
f(x)=\sum_{j=0}^{2^m-1} a_j x^{j},
\end{eqnarray*}
where $a_j\in \gf(2^m)$. The \emph{algebraic degree} $\mathrm{deg}(f)$ of $f$ is then equal to the maximum $2$-weight (or Hamming weight) of an
exponent $j$ for which $a_j \neq 0$, with the usual convention that the degree of the zero function is the negative  infinity.

For a nonempty proper  subset  $D$ of $\gf(2^m)$, a function $g \in \mathbb{B}_m$
is called an \emph{annihilator} of $D$ if $g f_D=0$. All annihilators of $D$
form an ideal of $\mathbb{B}_m$, denoted by $\mathrm{Ann}(D)$.
The  \emph{algebraic immunity  of  $D$}  is defined as
\begin{eqnarray*}
\ai (D)=\min \left \{ \mathrm{deg} (g) :  g\in \mathrm{Ann}(D) \setminus \{0\} \right \}.
\end{eqnarray*}
For convenience, we define $\ai (\emptyset)=- \infty$ and $\ai (\gf(2^m))=+ \infty$.
It is easy to see that $\ai(\cdot)$ is  monotone, which means that  $\ai(D_1)\le \ai (D_2)$ for any subsets  $D_1 \subseteq D_2$ of $\gf(2^m)$.

A vectorial Boolean $(m, r)$-function $F=(f_1, \cdots, f_r)$ is a function from $\gf(2^m)$ to $\gf(2)^r$.
For any vector $v=(v_1, \cdots, v_r) \in \gf(2)^r$, the component function $v\cdot F$ is the Boolean function given by
$v_1f_1 + \cdots + v_r f_r$.
The  \emph{algebraic immunity} of $F$ is defined as
\begin{eqnarray*}
\ai(F) = \min \left \{ \ai \left (F^{-1}(y) \right ): y \in \gf(2)^r \right \},
\end{eqnarray*}
where $F^{-1}(y)$ is  the \emph{preimage} of $y$ under $F$.
It was shown in \cite{DGM06} that the Hamming 
weight $\wt (f)$ of a Boolean function $f$  with prescribed algebraic immunity satisfies :
\begin{eqnarray*}
\sum_{i=0}^{\ai (f)-1} \binom{m}{i}  \le \wt (f) \le  \sum_{i=0}^{m-\ai (f)} \binom{m}{i}.
\end{eqnarray*}
It follows that $\ai (f) \le \lceil \frac{m}{2} \rceil$. 
Thus, Boolean functions attaining this upper bound are often said to have the optimal algebraic immunity.
For more information on vectorial Boolean functions, the reader is referred to \cite{Carlet}.

The \emph{$\tau$-th order nonlinearity $\mathrm{NL}_{\tau} (f)$} of a Boolean
function $f\in \mathbb B_m$ is the minimum \emph{Hamming distance}
$\mathrm{dist}(f,g)=\left |  \left \{ x\in \gf(2^m) : f(x)\neq g(x) \right \} \right |$  between $f$ and all functions $g$ of algebraic degree at most $\tau$.
The \emph{$\tau$-th order nonlinearity $\mathrm{NL}_{\tau} (F)$} of a vectorial function $F$ is the minimum $\tau$-th order nonlinearity of
its component functions. It was shown in \cite{Carlet08} that the $\tau$-th order nonlinearity of a vectorial $(m,r)$ function $F$ with given algebraic immunity
$\ai (F)=t$ satisfies
\begin{eqnarray}\label{eq:NL}
\mathrm{NL}_{\tau} (F) \ge \Upsilon_{m,r,t,\tau} ,
\end{eqnarray}
where $\Upsilon_{m,r,t,\tau}=2^{r-1} \sum_{i=0}^{t-\tau-1} \binom{m}{i} +2^{r-1} \sum_{i=t-2\tau}^{t-\tau-1} \binom{m-\tau}{i}$.
In the particular case that $r=\tau=1$, (\ref{eq:NL}) says that
\begin{eqnarray}\label{eq:NL-1-1}
\mathrm{NL}_{1} (f) \ge 2 \sum_{i=0}^{\ai(f)-2} \binom{m-1}{i},
\end{eqnarray}
where $f\in \mathbb B_m$.

\subsection{Minimal codes and cyclic codes}

We assume that the reader is familiar with the basics of linear codes (see for instance \cite{MS77} for detail). A linear code of
 length $n$ and dimension $k$ will be referred to as an $[n, k]$ code. Further,
 if the code has minimum distance $d$, it will be referred to as an $[n, k, d]$ code.

The Hamming weight (for short, weight) of a vector $\mathbf v$ is the number of its nonzero entries and is denoted $\mathrm{wt}(\mathbf v)$.
The minimum (respectively, maximum) weight of the code $\C$ is the minimum (respectively, maximum) nonzero weight of all codewords of $\C$,
$w_{\min} = \min (\mathrm{wt}(\mathbf c))$ (respectively, $w_{\max} = \max (\mathrm{wt}(\mathbf c))$).

Let $\bc=(c_0, \cdots, c_{n-1})$  be a codeword in $\C$.
The \emph{support} $\supp(\bc)$ of the codeword $\bc$  is the set
of indices of its nonzero coordinates:
$$\supp(\bc)=\{i: c_i \neq 0\}.$$
A codeword $\bc$ of the linear code $\C$ is called minimal if its support
does not contain the support of any other linearly independent codeword.
$\C$ is called a minimal linear code  if all codewords of $\C$  are minimal.
Minimal codes are a special class of linear codes.
A sufficient condition for a linear code to be minimal is given in the following lemma \cite{AB98}.
\begin{lemma}[Ashikhmin-Barg]\label{lem:AB}
 A linear code $\C$ over $\gf(q)$ is minimal if $\frac{w_{\min}}{w_{\max}} >\frac{q-1}{q}$.
\end{lemma}

Let $\mathcal  C$ be an $[n,k,d]$ linear code over $\mathrm{GF}(q)$ and $T$ a set of $t$
coordinate locations of $\mathcal C$.
Then the code $\C^T$ obtained from $\C$ by puncturing at the locations
 in $T$ is the code of length $n-t$ consisting of codewords of $\C$ which have their
 coordinate at the location $P$ deleted if $P\in T$ and left alone if $P \not \in T$,
 which is called the  \emph{punctured code}  of $\mathcal C$ on $T$.
The \emph{shortened code} $\mathcal C_{T}$ is the set of codewords from $\C$ that are zero at locations in $T$
, with coordinates in $T$ deleted.

An $[n,k]$ linear  code $\C$ over $\gf(q)$ is called \emph{cyclic} if $(c_0,c_1,\cdots, c_{n-1})\in \C$
implies that  the circular shift  $(c_{n-1},c_0,\cdots, c_{n-2})\in \C$. Clearly the vector
space $\gf(q)^n$ is isomorphic to the residue class ring $\gf(q)[X]/(X^n-1)$ (considered as an
additive group). An isomorphism is given by
\begin{eqnarray*}
(c_0, c_1, \cdots, c_{n-1}) \longleftrightarrow c_0+c_1X+\cdots +c_{n-1} X^{n-1}.
\end{eqnarray*}
From now on we do not distinguish between codewords of $\C$ and polynomials
 of degree  less than $n$ over $\gf(q)$. Note that the multiplication by $X$
 in $\gf(q)[X]/(X^n-1)$ amounts to the circular right shift
  $(c_0, c_1, \cdots, c_{n-1}) \longrightarrow (c_{n-1},c_0,\cdots, c_{n-2})$.
  From this it follows that a cyclic code $\C$ corresponds to an ideal in $\gf(q)[X]/(X^n-1)$, which we also denote by $\C$. Every $[n,k]$ cyclic code $\C$ over $\gf(q)$
is a principal ideal generated by some polynomial
$g(X)$ of degree $n-k$ that divides $X^n-1$. We shall call $g(X)$ and $h(X)=(X^n-1)/g(X)$ the \emph{generator polynomial}
and  the \emph{check polynomial} of $\C$, respectively. Note that
the codewords $g(X), Xg(X)$, $\cdots$, $X^{k-1}g(X)$
form a basis of $\C$.

Let us recall the \emph{BCH bound} on the minimum distance of cyclic codes \cite{HT72}.

\begin{theorem}
Let $h$ be an integer and $\delta$ be a positive integer  with $1\le \delta<n$.
Let $\alpha$
 be a primitive $n$-th root of unity in  the algebraic closure of $\gf(q)$.
Let $\C$ be a cyclic code of length $n$ over $\gf(q)$ with generator polynomial $g(X)$.
 If $g(X)$ has $\delta$ consecutive zeros $\alpha^{h}, \cdots, \alpha^{h+\delta-1}$,
 then the minimum distance of $\C$
 is greater than $\delta$.
\end{theorem}

\subsection{Reed-Muller codes}
Reed-Muller (RM) codes are classical codes that have enjoyed
unabated interest since their introduction in 1954 due to their simple recursive structure.

Let $\alpha$ be a primitive element of $\gf(2^m)$.
Let  $P_0=0$ and $P_j=\alpha^{j-1}$, where
$1\le j \le 2^m-1$.
Then $P_0,...,P_{2^m-1}$ is an enumeration of the points of the vector space $\gf(2^m)$.
Under this enumeration, the Reed-Muller code $\mathrm{RM}(\ell, m)$ of order $\ell $ in $m$ variables is defined as
\begin{eqnarray*}
\mathrm{RM}(\ell, m)=\left \{(f(P_0), \cdots, f(P_{2^m-1})):  f\in \mathbb{B}_m, \mathrm{deg}(f)\le \ell\right \}.
\end{eqnarray*} 
In this paper, we index the coordinates of the code $\mathrm{RM}(\ell, m)$ with the sequence 
$(P_0, P_1, \ldots, P_{2^m-1})$. 
The general affine group over  $\gf(2^m)$, denoted by $\mathrm{GA}(1,2^m)$, is defined by
$$\mathrm{GA}(1,2^m) =\left \{\pi_{a,b} : a\in  \gf(2^m)^*, b\in \gf(2^m) \right \}, $$
where $\pi_{a,b}$ is the permutation on $\gf(2^m)$ defined by
$x \mapsto ax+b$. Since $\mathrm{deg}(f(x))=\mathrm{deg}(f(ax+b))$ for any $(a, b) \in \gf(2^m)^* \times \gf(2^m)$,
the Reed-Muller code $\mathrm{RM}(\ell, m)$ is invariant under the action by
 $\mathrm{GA}(1,2^m)$.
 We denote the codes  obtained after the  \emph{puncturing} and \emph{shortening} operation on $\mathrm{RM}(\ell, m)$ at
 the coordinate location $P_0$
   as $\mathrm{PRM}(\ell, m)$ and $\srm (\ell,m)$, respectively.
It is easy to see that
the punctured code $\mathrm{PRM}(\ell, m)$ and the shortened code $\srm (\ell,m)$ of the Reed-Muller code
$\mathrm{RM}(\ell, m)$  are cyclic codes of length $2^m-1$. Let $g_{\ell, \alpha}(X)$ and $g_{\ell, \alpha}^*(X)$  denote the
generator polynomials of the cyclic codes $\prm (\ell, m)$ and $\srm (\ell, m)$, respectively.

 The following proposition describes the generator polynomials of the punctured Reed-Muller codes,
 which are not hard to prove \cite{AssKey98}.
 \begin{proposition}\label{prop:prm-generator-alpha}
The punctured Reed-Muller code  $\mathrm{PRM} (\ell,m)$ is a cyclic code of
dimension $\sum_{j=0}^{\ell} \binom{m}{j}$
with  generator polynomial
 \begin{eqnarray*}
 g_{\ell, \alpha}(X)= \prod_{ \scriptsize{\begin{array}{c}0<i_{m-1}+\cdots+ i_0 \le m-1 -\ell
 \\ i_{m-1}, \cdots, i_0 \in \{0,1\} \end{array}}} (X-\alpha^{i_{m-1}2^{m-1}+ \cdots+i_0 2^0}).
 \end{eqnarray*}
 \end{proposition}

The following proposition is taken from Corollary 4 of \cite{Assmus92}.
 \begin{proposition}\label{prop:mini-geometry}
The minimum weight of the Reed-Muller code $\mathrm{RM}(\ell,m)$ is $2^{m-\ell}$ and
the minimum-weight codewords are the incidence vectors of the $(m-\ell)$-flats
of the affine space $\mathrm{AG}(m,2)$ of dimension $m$ over $\gf(2)$.
The minimum weight of the punctured code $\mathrm{PRM} (\ell,m)$ is $2^{m-\ell}-1$ and the
 minimum-weight codewords are the incidence vectors of the $(m-\ell-1)$-dimensional subspaces of the projective space $\mathrm{PG}(m-1,2)$
 of dimension $m-1$ over $\gf(2)$.
 \end{proposition}

 \begin{lemma}\label{lem:minimum-set-geometry}
The minimum weight of the shortened Reed-Muller code $\srm(\ell,m)$ is $2^{m-\ell}$ and
the minimum-weight codewords are the incidence vectors of the $(m-\ell)$-flats not passing through the origin
in $\mathrm{AG}(m,2)$.
 \end{lemma}
\begin{proof}
Note that the shortened code $\srm(\ell,m)$  consists of codewords of $\mathrm{RM}(\ell,m)$  that are zero at
the origin of $\mathrm{AG}(m,2)$.
The desired  result then follows  from Proposition \ref{prop:mini-geometry}.
\end{proof}

 \section{Minimal codes from sets without nonzero low-degree annihilators }\label{sec:codes-sets}
 
In this section we present a general construction of minimal codes
using subcodes of Reed-Muller codes and subsets of finite fields
 without nonzero low-degree annihilators.

Here and hereafter, for any subset $D$ of $\gf(2^m)$, let $D^*$ denote the set $D\setminus \{0\}$ and $\overline{D}$ stand for the complement of $D$ in $\gf(2^m)$.
In particular, if $D\subseteq \gf(2^m)^*$, then $\overline{D}^*$ is the complement of $D$ in $\gf(2^m)^*$.
Let $\mathbb B_m^0$ denote the set $\{f \in \mathbb{B}_m : f(0)=0\}$.

The following theorem  presents a general approach to constructing binary minimal codes, and produces many classes
of binary minimal codes by selecting some subcodes of  Reed-Muller codes and subsets of $ \gf(2^m)$ with special annihilators.

 \begin{theorem}\label{thm:min-rm-ann}
 Let $\C$ be a $k$-dimensional subcode of the  Reed-Muller code $\rm{RM} (\ell, m)$. Let $D$ be a  subset of $\gf(2^m)$.
 Then
 $\C^{\overline{D}}$ is a  minimal code of dimension $k$ if and only if the following two conditions hold:
\begin{enumerate}[(1)]
\item the code $\C$ is  minimal, and
\item for any two nonzero codewords $(f_1(P_0), \cdots, f_1(P_{2^m-1}))$ and
$(f_2(P_0), \cdots, f_2(P_{2^m-1}))$ of $\C$ (including the case $f_1=f_2$), where $f_1, f_2 \in \mathbb{B}_m$, the product $f_1f_2$ of $f_1$
and $f_2$
is not an annihilator of $D$.
\end{enumerate}
 \end{theorem}
 
\begin{proof}
Let $\C$ be a linear code satisfying Conditions (1) and (2). Let $(f(P_0), \cdots, f(P_{2^{m}-1}))$
be any nonzero codeword of $\C$, where $f \in \mathbb B_{m}$.
Let $f_1=f_2=f$. Then $f_1f_2=f^2=f$, and by Condition (2), $f$ is not an annihilator of $D$, i.e.,  
$\left ( f(P) \right )_{P\in D} \neq 0$.
Consequently the punctured code $\C^{\overline{D}}$ has the same dimension as the original code $\C$.
Suppose that $\C^{\overline{D}}$ is not minimal. Then there exist two distinct nonzero codewords $\left ( f_1(P) \right )_{P\in D}, \left ( f_2(P) \right )_{P\in D}  \in \C^{\overline{D}}$, where  $f_1, f_2 \in
\mathbb B_m$,  such that $\supp \left ( \left ( f_1(P) \right )_{P\in D}  \right ) \subsetneq \supp \left ( \left ( f_2(P) \right )_{P\in D}  \right )$.
This clearly forces 
$$\supp\left ( \left ( f_2(P) \right )_{P\in D}  \right )  = \supp \left ( \left ( (f_1 + f_2)(P) \right )_{P\in D}  \right ) \dot\cup
 \supp \left ( \left ( f_1(P) \right )_{P\in D}  \right ).$$  
It follows that $f_1 (f_1 + f_2) \in \mathrm{Ann}(D)$, which is contrary to Condition (2).
 Therefore $\C^{\overline{D}}$ is minimal.

Conversely, assume $\C^{\overline{D}}$ is a minimal code with dimension $k$. It is clear that $\C$ is minimal.
It remains to show that Condition (2) holds. On the contrary, suppose that
there exist  two nonzero codewords $(f_1(P_0), \cdots, f_1(P_{2^m-1}))$ and
$(f_2(P_0), \cdots, f_2(P_{2^m-1}))$ of $\C$, where $f_1, f_2 \in \mathbb{B}_m$, such that  $f_1f_2 \in \mathrm{Ann}(D)$.
 Then $\supp \left ( \left ( f_1(P) \right )_{P\in D}  \right )$ and  $\supp \left ( \left ( f_2(P) \right )_{P\in D}  \right )$
 are disjoint. This yields 
  $$ \supp \left ( \left ( f_1(P) \right )_{P\in D}  \right ) \subsetneq   \supp \left ( \left ( (f_1 + f_2)(P) \right )_{P\in D}  \right),$$ 
 which contradicts the minimality of nonzero codewords of $\C^{\overline{D}}$. This completes the proof.
\end{proof}

 \begin{corollary}\label{cor:ad>2l}
 Let $\C $ be a minimal code contained in the  Reed-Muller code $ \mathrm{RM} (\ell, m)$.
 Let $D$ be a  subset of $\gf(2^m)^*$ with $\ai (D) \ge  2 \ell +1$.
 Then $\C^{\overline{D}}$ is a  minimal code of dimension $k$,
 where  $k$ equals  the dimension of $\C$.
 \end{corollary}
 \begin{proof}
 Let $(f_1(P_0), \cdots, f_1(P_{2^m-1}))$ and
$(f_2(P_0), \cdots, f_2(P_{2^m-1}))$ be any two nonzero codewords of $\C$, where $f_1, f_2 \in \mathbb{B}_m$.
Since $\C$ is minimal, $f_1f_2$ cannot be  the zero function.
From $\mathrm{deg}(f_1f_2)\le 2\ell$ and $\ai (D) \ge  2 \ell +1$, we conclude that $f_1f_2 \not \in \mathrm{Ann}(D)$.
The desired result then follows from Theorem \ref{thm:min-rm-ann}.
 \end{proof}

  \begin{corollary}\label{cor:C-shortened-ad>2l}
 Let $\C $ be a minimal code contained in the  shortened Reed-Muller code $ \mathrm{SRM} (\ell, m)$.
 Let $D$ be a  subset of $\gf(2^m)^*$ with $\ai (D \cup \{ 0 \}) \ge  2 \ell +1$.
 Then $\C^{\overline{D}^*}$ is a  minimal code of dimension $k$,
 where  $k$ equals  the dimension of $\C$.
 \end{corollary}
 \begin{proof}
Denote by $\C'$ the codes $\{ (0, \mathbf c): \mathbf c \in \C \}$.
Since $\ai (D \cup \{ 0 \}) \ge  2 \ell +1$, it follows from Corollary \ref{cor:ad>2l} that $\C'^{\overline{D \cup \{0\}}}$
is a minimal code of dimension $k$. The desired conclusion  then follows from
the definitions of $\C'$ and shortened codes.
 \end{proof}

 To deduce a lower bound on the minimum distance of the codes from sets with high algebraic immunity,
we need  some additional lemmas.
Denote by $\mathrm{Ann}_{t-1} (g)$ the vector space of those annihilators of degrees at most $t-1$ of $\supp(g)$.
\begin{lemma}\label{lem:wt-dim}
Let $D\subseteq \gf(2^m)$ with $\ai (D)=t$ and $g\in \mathbb B_m$. Then
\begin{eqnarray*}
\wt (g f_D) \ge \mathrm{dim} \left (\mathrm{Ann}_{t-1} (1+g) \right ).
\end{eqnarray*}
\end{lemma}

\begin{proof}
Let $w= \wt (g f_D)$ and let $\mathcal Q$ be the set consisting of $w$ distinct points $Q_1$,
$\cdots$, $Q_w$ in $D$ satisfying $g(Q_i)=1$.
Consider the evaluation map $\mathrm{Ev}_{\mathcal Q}$
$$\mathrm{Ann}_{t-1} (1+g)  \longrightarrow \gf(2)^w, $$
defined by $\mathrm{Ev}_{\mathcal Q} (h)= \left ( h(Q_1), \cdots, h(Q_w) \right )$.
Then $\mathrm{Ev}_{\mathcal Q}$ is a linear transformation.
Suppose the assertion of the lemma is false. Then $\mathrm{Ev}_{\mathcal Q}$ is not injective.
Thus there exists a nonzero function $h$ in $\mathrm{Ann}_{t-1} (1+g) $ such that $h f_{\mathcal Q}=0$.
It follows easily  that $h \in \mathrm{Ann}(D)$, which contradicts the condition that $\ai (D)=t$.
This completes the proof.
\end{proof}

Little is known about the behavior of the annihilators
of a polynomial of a given degree. Mesnager \cite{Sihem08} proved the following lower bound on
the dimension of $\mathrm{Ann}_{t-1} (g)$
$$\mathrm{dim} \left (\mathrm{Ann}_{t-1} (g) \right ) \ge \sum_{i=0}^{t-\tau-1} \binom{m-\tau}{i},$$
where $g \in \mathbb B_m$ with $\mathrm{deg}(g)=\tau$.
Lemma \ref{lem:wt-dim} indicates that the following lemma holds.
\begin{lemma}\label{lem:wt-bound}
Let $D\subseteq \gf(2^m)$ with $\ai (D)=t$ and $g\in \mathbb B_m \setminus \{0\}$ with $\mathrm{deg} (g)=\tau$. Then
\begin{eqnarray*}
\wt (g f_D) \ge \sum_{i=0}^{t-\tau-1} \binom{m-\tau}{i}.
\end{eqnarray*}
\end{lemma}

 A method of explicitly constructing minimal codes by puncturing the Simplex codes  is given in the following theorem.
 \begin{theorem}\label{thm:C-D}
 Let $D$ be a  subset of $\gf(2^m)^*$ with $\ai \left (D\cup \{0\} \right )= t \ge  3$.
 Let $\C(D)$ be the linear code given by
 \begin{eqnarray}\label{eqn:C-D}
 \C(D) = \left \{ \left ( \tr^m_1 (a x) \right )_{x \in D} : a \in \gf (2^m)\right \}.
 \end{eqnarray}
Then $\C (D)$ is a  minimal code with parameters $\left [ |D|, m, \ge \sum_{i=0}^{t-2} \binom{m-1}{i} \right ]$.
 \end{theorem}
 \begin{proof}
 Let $\C \subseteq \srm (1, m)$ be the Simplex code defined by
 $$\left \{  \left ( \tr^m_1 (a P_1) , \cdots, \tr^m_1 (a P_{2^m-1})  \right ) : a \in \gf (2^m)  \right \}.$$
 It is well-known that $\C$ is a minimal code of dimension $m$.
 The desired conclusions are immediate from Corollary \ref{cor:C-shortened-ad>2l} and Lemma \ref{lem:wt-bound}.
 \end{proof}

 \begin{corollary}\label{cor:C-Supp-f}
 Let $m\ge 5$ be an integer and let $f\in \mathbb B_m$ be a  Boolean function with $\ai \left (f \right )=t \ge  3$ and $\wt (f) \ge 2^{m-1}$.
 Let $D= \supp (f) \setminus \{0\}$.
Then the code $\C \left (D \right )$ defined by (\ref{eqn:C-D}) is an   $m$-dimensional minimal code with
minimum distance $d$ satisfying
\begin{eqnarray*}
d \ge     \sum_{i=0}^{t-2} \binom{m-1}{i}  + \frac{1}{2}\left ( \wt (f)- 2^{m-1}  \right ).
\end{eqnarray*}
 \end{corollary}
 
 \begin{proof}
 It follows from Theorem \ref{thm:C-D} that $\C \left (D \right )$ is an $m$-dimensional minimal code.
 It remains to prove the lower bound on the minimum distance of $\C \left (D \right )$.
 Denote by $g$  the Boolean function $\tr^m_1 (a x)$, where $a \in \gf (2^m)^*$.
 Let $w=  \wt \left ( \left ( \tr^m_1 (a x) \right )_{x \in D} \right )$.
 Thus $$w= \left| \{x\in \gf (2^m): f(x)=1, g(x)=1\} \right |. $$
 By the definition of the Hamming distance between $f$ and $g$, we have
 \begin{eqnarray*}
 \begin{array}{rl}
 \dist (f,g)=&\left| \{x\in \gf (2^m): f(x)=0, g(x)=1 \}\right |\\
             & + \left| \{x\in \gf (2^m): f(x)=1, g(x)=0 \}\right |\\
            =& \wt (g) -w +\wt (f) -w\\
            =& \wt (g) +\wt (f) -2w.
 \end{array}
 \end{eqnarray*}
This yields 
\begin{eqnarray}\label{eq:local-global}
\begin{array}{rl}
w=& \left ( \wt (g) +\wt (f) -   \dist (f,g)  \right )/2\\
=& \left (    \dist (f,1+g)+ \wt (g)  + \wt (f)   -2^m  \right )/2.\\
\end{array}
\end{eqnarray}
The desired conclusion then follows from (\ref{eq:NL-1-1}), (\ref{eq:local-global}) and the fact that $\wt (g) =2^{m-1}$.
 \end{proof}
Now, consider balanced Boolean functions in Corollary \ref{cor:C-Supp-f}.Then, we obtain the following  result.
\begin{corollary}\label{cor:C-Supp-f-balanced}
 Let $m\ge 5$ be an integer. Let $f\in \mathbb B_m$ be a balanced   Boolean function with $\ai \left (f \right )=t \ge  3$.
 Let $D= \supp (f) \setminus \{0\}$.
Then the code $\C \left (D \right )$ defined by (\ref{eqn:C-D}) is a   $\left [2^{m-1}, m, \ge  \sum_{i=0}^{t-2} \binom{m-1}{i} \right ]$ minimal code.
Moreover, if $f$ has
optimum algebraic immunity, then $\C \left (D \right )$
is a minimal code with parameters $\left [ 2^{m-1}, m, \ge  \sum_{i=0}^{  \lceil \frac{m-4}{2} \rceil } \binom{m-1}{i}\right].$
 \end{corollary}

 In order to apply Theorem \ref{thm:min-rm-ann} to construct minimal codes, finding sets
 without low-degree nonzero annihilators is very important. Let $\alpha$ be a primitive element of  $\gf(2^m)^*$,
 $h$ and $\delta$ be two  integers with $\delta >0$. Denote then
 $[h; \delta]_{\alpha}=\left \{\alpha^{h}, \alpha^{h+1}, \cdots, \alpha^{h+\delta-1} \right \}$.

\begin{lemma}\label{lem:ai-det-0}
Let $\delta $ be an integer with $\sum_{i=0}^{t} \binom{m}{i} \le \delta  < \sum_{i=0}^{t+1} \binom{m}{i}$.
Then $\ai \left ( [h; \delta]_{\alpha} \right )= t+1$.
\end{lemma}
\begin{proof}
Let $f$ be a  function of degree at most $t$ in  $ \mathrm{Ann} ([h; \delta]_{\alpha} )$ and
$\tilde{f}(X) \in \prm (t, m)$ be the codeword 
associated with $f$. By assumption, we see that
\begin{eqnarray}\label{eq:f-codew}
\tilde{f}(X)=X^h \sum_{i= \delta }^{2^m-2} c_i X^i,
\end{eqnarray}
where $c_i \in \gf(2)$. By the definition of the generator polynomial $g_{t, m}(X)$
of $\prm (t, m)$, the codeword $X^{-h} \tilde{f}(X)$ can also be written as
\begin{eqnarray}\label{eq:f-codew-g}
X^{-h} \tilde{f}(X) =\left  (a_0+a_1 X+ \cdots a_{\delta_t -1 } X^{\delta_t-1} \right ) g_{t, m}(X),
\end{eqnarray}
where $\delta_t= \sum_{i=0}^{t} \binom{m}{i} $ and $a_i \in \gf(2)$.
Combining (\ref{eq:f-codew}) with (\ref{eq:f-codew-g}) yields $a_0= \cdots = a_{\delta_t -1 } =0$.
We thus get $f=0$. Hence $\ai ([h; \delta]_{\alpha} ) \ge t+1$.

Let  $f$ be the  function corresponding to the codeword $\tilde{f}(X) $ of $\prm (t+1, m)$, where $\tilde{f}(X) $
is  given by
\begin{eqnarray*}
\tilde{f}(X)=X^h
\left (X^{\delta}+ X^{\delta+1}  \cdots +  X^{\delta_{t+1}-1} \right ) g_{t+1, m}(X),
\end{eqnarray*}
where $\delta_{t+1}= \sum_{i=0}^{t+1} \binom{m}{i} $.
It is easy to check that $f \in \mathrm{Ann} ( [h; \delta]_{\alpha} ) \setminus \{0\}$.
It follows that $\ai ([h; \delta]_{\alpha} ) \le t+1$.

Summarising the discussions above yields $\ai ([h; \delta]_{\alpha} ) = t+1$, which is the desired conclusion. 
\end{proof}

The proof of the following lemma is similar to the proof of Lemma \ref{lem:ai-det-0}, with punctured Reed-Muller codes replaced  by
shortened Reed-Muller codes,  and therefore is omitted.

 \begin{lemma}\label{lem:[h;delta]+0}
Let $\delta $ be an integer with $\sum_{i=1}^{t} \binom{m}{i} \le \delta  < \sum_{i=1}^{t+1} \binom{m}{i}$.
Then $\ai \left ( [h; \delta]_{\alpha}  \cup \{0\}\right )= t+1$.
\end{lemma}

Now we recall some facts on Gauss sums which will be needed to derive  an improved  lower bound
on the minimum distance of the codes from the sets $[h; \delta]_{\alpha}$.
Let $\xi_{q-1}$ denote the complex primitive $(q-1)$th root of unity $e^{2\pi \sqrt{-1}/(q-1)}$. Let
$\alpha $ be a primitive element of $\gf(q)$, and
let $\chi$ be the character of $\gf(q)^*$ given
by
\begin{eqnarray*}
\chi(\alpha^j)=\xi_{q-1}^j,
\end{eqnarray*}
where $0 \le j \le  q-2$. The \emph{Gauss sum} associated to $\chi^j$ over $\gf(q)$ with $q=2^m$ is defined by
\begin{eqnarray*}\label{eq:Gass-number}
G(\chi^j)= \sum_{i=0}^{q-2} (-1)^{\tr^{m}_1\left( \alpha^i \right ) } \chi\left (\alpha^{ij}\right ), \text{ for } j=0, \cdots, q-2.
\end{eqnarray*}
Then $G(\chi^0)=-1$ and $G(\chi^j)$ ($1\le j \le q-2$) satisfies the fundamental property \cite[p. 132]{IreRosen72}
\begin{eqnarray}\label{eq:sqrt-q}
G(\chi^j) \overline{G(\chi^j)} =q,
\end{eqnarray}
where the bar denotes complex conjugate.
It is sometimes convenient to view the Gauss sum $G(\chi^j)$ as a function of $\chi^j$.
This amounts to viewing $\chi^j \mapsto G(\chi^j)$ as the multiplicative Fourier transformation
of the function $(-1)^{\tr^m_1(x)}$ on $\gf(q)^*$. The following 
Fourier inversion formula allows us to recover $(-1)^{\tr^m_1(x)}$ from $G(\chi^j)$ by
\begin{eqnarray}\label{eq:add-mulp-Gauss}
(-1)^{\tr^m_1(\alpha^i)}=\frac{1}{q-1} \sum_{j=0}^{q-2}  \overline{\chi}^j (\alpha^i)  G(\chi^j) .
\end{eqnarray}

We will need the following lemma, whose proof can be found in \cite{CFeng09}.
\begin{lemma}\label{lem:exp-sin-ln}
Let $q=2^m$. It holds
\begin{eqnarray*}
\sum_{j=1}^{2^{m-1}-1}  \frac{1}{\sin \left (\pi j/(q-1) \right )} \le   \frac{q-1}{2 \pi } \ln \left ( \frac{4(q-1)}{\pi} \right ).
\end{eqnarray*}

\end{lemma}

Combining Simplex codes with the sets $[h; \delta]_{\alpha}$, an infinite class of binary minimal codes is given in the following theorem.
\begin{theorem}\label{thm:C-[h;delta]}
Let $q=2^m$. Let $\delta $ be an integer with $\sum_{i=1}^{t-1} \binom{m}{i}  \le \delta  < \sum_{i=1}^{t} \binom{m}{i}$  and $3\le t \le m$.
Then the code $\C \left ([h; \delta]_{\alpha}  \right )$ defined by (\ref{eqn:C-D}) is a   minimal code with parameters $[\delta, m ,d]$, where
\begin{eqnarray*}
d \ge \max \left \{ \sum_{i=0}^{t-2} \binom{m-1}{i}, \frac{\delta-1}{2}- \frac{\sqrt{q}}{2 \pi}  \ln \left ( \frac{4(q-1)}{\pi} \right ) \right \}.
 \end{eqnarray*}

 \end{theorem}
\begin{proof}
Let $\bc=\left ( \tr^m_1 (\lambda \alpha^{h}), \cdots, \tr^m_1 (\lambda \alpha^{h+\delta-1} ) \right ) $
be a nonzero codeword of $\C \left ([h; \delta]_{\alpha}  \right )$.
Then
\begin{eqnarray}\label{eq:wt-av-F}
\begin{array}{rl}
\wt (\bc) = & \frac{1}{2} \sum_{ i=h }^{h+\delta-1} \left (1-(-1)^{ \tr^m_1 (\lambda \alpha^{i})} \right )\\
=& \frac{1}{2}\delta-\frac{1}{2}  \sum_{ i=0 }^{\delta -1}  (-1)^{ \tr^m_1 (\lambda \alpha^h \alpha^{i})}\\
=& \frac{1}{2} \delta- \frac{1}{2}  F,
\end{array}
\end{eqnarray}
where $F=\sum_{ i=0 }^{\delta -1}  (-1)^{ \tr^m_1 (\lambda \alpha^h \alpha^{i})}$.
Set $\lambda'= \lambda \alpha^h$. Substituting (\ref{eq:add-mulp-Gauss}) into $F$ yields
\begin{eqnarray*}
\begin{array}{rl}
|F|=& \frac{1}{q-1}  |\sum_{i=0}^{\delta -1}   \sum_{j=0}^{q-2} \overline{\chi}
(\lambda'^j \alpha^{ij})  G(\chi^{j})| \\
=& \frac{1}{q-1}  |\sum_{j=0}^{q-2}  \overline{\chi} (\lambda'^j ) G(\chi^{j}) \sum_{i=0}^{\delta -1}  \overline{\chi}
( \alpha^{ij}) | \\
=&|-\frac{1}{q-1} \delta + \frac{1}{q-1}  \sum_{j=1}^{q-2}  \overline{\chi} (\lambda'^j ) G(\chi^{j}) \sum_{i=0}^{\delta -1}  \overline{\chi}
( \alpha^{ij})  |\\
\le  & \frac{1}{q-1} \delta +\frac{\sqrt{q}}{q-1} \sum_{j=1}^{q-2} \left |\sum_{i=0}^{\delta -1}  \xi_{q-1}^{ij} \right |,
\end{array}
\end{eqnarray*}
where the last inequality follows from (\ref{eq:sqrt-q}).
A simple calculation yields
\begin{eqnarray}\label{eq:F-indpendentOFdelta}
\begin{array}{rl}
|F|\le & \frac{1}{q-1} \delta +\frac{\sqrt{q}}{q-1} \sum_{j=1}^{q-2} \left | \frac{\sin \left (\pi \delta j/(q-1) \right )}{\sin \left (\pi j/(q-1) \right )} \right | \\
\le & \frac{1}{q-1} \delta +\frac{2\sqrt{q}}{q-1} \sum_{j=1}^{2^{m-1}-1}  \frac{1}{\sin \left (\pi j/(q-1) \right )}\\
\le &  \frac{1}{q-1} \delta + \frac{\sqrt{q}}{\pi}  \ln \left ( \frac{4(q-1)}{\pi} \right ),
\end{array}
\end{eqnarray}
where the last inequality follows from Lemma \ref{lem:exp-sin-ln}.
Combining (\ref{eq:wt-av-F}) with (\ref{eq:F-indpendentOFdelta}), we deduce that
\begin{eqnarray*}
\wt (\bc) \ge \frac{\delta-1}{2}- \frac{\sqrt{q}}{2 \pi}  \ln \left ( \frac{4(q-1)}{\pi} \right ).
\end{eqnarray*}
The desired conclusion then follows from Theorem \ref{thm:C-D} and Lemma \ref{lem:[h;delta]+0}.

\end{proof}

\begin{figure}
\centering
\begin{tikzpicture}
	\begin{axis}[
	    xmin=36,
	    xtick={36,76,...,254},
	    minor x tick num=1,
	    ymin=8,
	    ytick={8,28,...,128},
	    minor y tick num=1,
	    xmax=254,
	    ymax=127,
	    xlabel=Length of Code,
	    ylabel=Lower Bound of Minimum Distance,
		height=8cm,
		width=8cm,
		grid=major,
         grid style={line width=.2pt, draw=gray!50},
          legend style={at={(0.29,0.96)},
                       anchor=north,legend columns=1},
	]
	\addplot[style={densely dashed},line width=1.2pt,red] coordinates {
(36,8)  (37,8)  (38,8)  (39,8)  (40,8)  (41,8)  (42,8)  (43,8)
(44,8)  (45,8)  (46,8)  (47,9)  (48,9)  (49,10)  (50,10)  (51,11)
(52,11)  (53,12)  (54,12)  (55,13)  (56,13)  (57,14)  (58,14)  (59,15)
(60,15)  (61,16)  (62,16)  (63,17)  (64,17)  (65,18)  (66,18)  (67,19)
(68,19)  (69,20)  (70,20)  (71,21)  (72,21)  (73,22)  (74,22)  (75,23)
(76,23)  (77,24)  (78,24)  (79,25)  (80,25)  (81,26)  (82,26)  (83,27)
(84,27)  (85,28)  (86,28)  (87,29)  (88,29)  (89,30)  (90,30)  (91,31)
(92,31)  (93,32)  (94,32)  (95,33)  (96,33)  (97,34)  (98,34)  (99,35)
(100,35)  (101,36)  (102,36)  (103,37)  (104,37)  (105,38)  (106,38)  (107,39)
(108,39)  (109,40)  (110,40)  (111,41)  (112,41)  (113,42)  (114,42)  (115,43)
(116,43)  (117,44)  (118,44)  (119,45)  (120,45)  (121,46)  (122,46)  (123,47)
(124,47)  (125,48)  (126,48)  (127,49)  (128,49)  (129,50)  (130,50)  (131,51)
(132,51)  (133,52)  (134,52)  (135,53)  (136,53)  (137,54)  (138,54)  (139,55)
(140,55)  (141,56)  (142,56)  (143,57)  (144,57)  (145,58)  (146,58)  (147,59)
(148,59)  (149,60)  (150,60)  (151,61)  (152,61)  (153,62)  (154,62)  (155,63)
(156,63)  (157,64)  (158,64)  (159,65)  (160,65)  (161,66)  (162,66)  (163,67)
(164,67)  (165,68)  (166,68)  (167,69)  (168,69)  (169,70)  (170,70)  (171,71)
(172,71)  (173,72)  (174,72)  (175,73)  (176,73)  (177,74)  (178,74)  (179,75)
(180,75)  (181,76)  (182,76)  (183,77)  (184,77)  (185,78)  (186,78)  (187,79)
(188,79)  (189,80)  (190,80)  (191,81)  (192,81)  (193,82)  (194,82)  (195,83)
(196,83)  (197,84)  (198,84)  (199,85)  (200,85)  (201,86)  (202,86)  (203,87)
(204,87)  (205,88)  (206,88)  (207,89)  (208,89)  (209,90)  (210,90)  (211,91)
(212,91)  (213,92)  (214,92)  (215,93)  (216,93)  (217,94)  (218,99)  (219,99)
(220,99)  (221,99)  (222,99)  (223,99)  (224,99)  (225,99)  (226,99)  (227,99)
(228,99)  (229,100)  (230,100)  (231,101)  (232,101)  (233,102)  (234,102)  (235,103)
(236,103)  (237,104)  (238,104)  (239,105)  (240,105)  (241,106)  (242,106)  (243,107)
(244,107)  (245,108)  (246,120)  (247,120)  (248,120)  (249,120)  (250,120)  (251,120)
(252,120)  (253,120)  (254,127)
	};
	\addlegendentry{\tiny{Lower Bound of Theorem \ref{thm:C-[h;delta]}}}

	\addplot[style={densely dashed},line width=0.9pt,blue] coordinates {
(36,8)  (37,8)  (38,8)  (39,8)  (40,8)  (41,8)  (42,8)  (43,8)
(44,8)  (45,8)  (46,8)  (47,8)  (48,8)  (49,8)  (50,8)  (51,8)
(52,8)  (53,8)  (54,8)  (55,8)  (56,8)  (57,8)  (58,8)  (59,8)
(60,8)  (61,8)  (62,8)  (63,8)  (64,8)  (65,8)  (66,8)  (67,8)
(68,8)  (69,8)  (70,8)  (71,8)  (72,8)  (73,8)  (74,8)  (75,8)
(76,8)  (77,8)  (78,8)  (79,8)  (80,8)  (81,8)  (82,8)  (83,8)
(84,8)  (85,8)  (86,8)  (87,8)  (88,8)  (89,8)  (90,8)  (91,8)
(92,29)  (93,29)  (94,29)  (95,29)  (96,29)  (97,29)  (98,29)  (99,29)
(100,29)  (101,29)  (102,29)  (103,29)  (104,29)  (105,29)  (106,29)  (107,29)
(108,29)  (109,29)  (110,29)  (111,29)  (112,29)  (113,29)  (114,29)  (115,29)
(116,29)  (117,29)  (118,29)  (119,29)  (120,29)  (121,29)  (122,29)  (123,29)
(124,29)  (125,29)  (126,29)  (127,29)  (128,29)  (129,29)  (130,29)  (131,29)
(132,29)  (133,29)  (134,29)  (135,29)  (136,29)  (137,29)  (138,29)  (139,29)
(140,29)  (141,29)  (142,29)  (143,29)  (144,29)  (145,29)  (146,29)  (147,29)
(148,29)  (149,29)  (150,29)  (151,29)  (152,29)  (153,29)  (154,29)  (155,29)
(156,29)  (157,29)  (158,29)  (159,29)  (160,29)  (161,29)  (162,64)  (163,64)
(164,64)  (165,64)  (166,64)  (167,64)  (168,64)  (169,64)  (170,64)  (171,64)
(172,64)  (173,64)  (174,64)  (175,64)  (176,64)  (177,64)  (178,64)  (179,64)
(180,64)  (181,64)  (182,64)  (183,64)  (184,64)  (185,64)  (186,64)  (187,64)
(188,64)  (189,64)  (190,64)  (191,64)  (192,64)  (193,64)  (194,64)  (195,64)
(196,64)  (197,64)  (198,64)  (199,64)  (200,64)  (201,64)  (202,64)  (203,64)
(204,64)  (205,64)  (206,64)  (207,64)  (208,64)  (209,64)  (210,64)  (211,64)
(212,64)  (213,64)  (214,64)  (215,64)  (216,64)  (217,64)  (218,99)  (219,99)
(220,99)  (221,99)  (222,99)  (223,99)  (224,99)  (225,99)  (226,99)  (227,99)
(228,99)  (229,99)  (230,99)  (231,99)  (232,99)  (233,99)  (234,99)  (235,99)
(236,99)  (237,99)  (238,99)  (239,99)  (240,99)  (241,99)  (242,99)  (243,99)
(244,99)  (245,99)  (246,120)  (247,120)  (248,120)  (249,120)  (250,120)  (251,120)
(252,120)  (253,120)  (254,127)
	};
	\addlegendentry{\tiny{Lower Bound of Theorem \ref{thm:C-D}}}

		\end{axis}
\end{tikzpicture}
\caption{A Comparison of Lower Bounds of Theorems \ref{thm:C-D} and \ref{thm:C-[h;delta]} } \label{fig:lower bound}
\end{figure}

\begin{corollary}\label{cor:[m(m+1)/2,m]}
Let $m\ge 5$ be an integer and $\alpha$ a primitive element of $\gf (2^m)$.
Let $\C$ be the set given by
\begin{eqnarray*}
\C=\left \{ \left ( \tr^m_1 \left (a \alpha^{0} \right ), \cdots, \tr^m_1 \left (a \alpha^{m(m+1)/2-1} \right )  \right ): a \in \gf(2^m) \right \}.
\end{eqnarray*}
Then $\C$ is a binary minimal code of dimension $m$ and length $m(m+1)/2$.
\end{corollary}

According to our best knowledge there is only one known example of binary minimal codes with dimension $m$  and length $m(m+1)/2$
, which was introduced in \cite{ZYW19}. These minimal codes can be described as follows:
\begin{eqnarray}\label{eq:P and P+P}
\left \{
\begin{array}{c}
 \left ( \tr^m_1(a \alpha_1), \cdots, \tr^m_1(a \alpha_m), \right . \\
\left .  \tr^m_1(a (\alpha_1+\alpha_2)), \cdots, \tr^m_1(a (\alpha_{m-1}+\alpha_m))    \right )
 \end{array}
 : a \in \gf (2^m)  \right \},
\end{eqnarray}
where $\alpha_1, \cdots, \alpha_m$
form a basis of $\gf(2^m)$ over $\gf(2)$.
Obviously, the codes obtained in (\ref{eq:P and P+P}) are unique, up to equivalence.
Their minimum   distance $d(m)$  is equal to $m$.  Many infinite families of minimal codes of
dimension $m$ and length $m(m+1)/2$
can be produced form
Corollary \ref{cor:[m(m+1)/2,m]}.
Denote by $d_{max}(m)$ and $d_{min}(m)$ the  largest and smallest values of the minimum distances of the codes of Corollary \ref{cor:[m(m+1)/2,m]}, respectively.
Figure \ref{fig:minimum distance} shows that the minimum distance of the minimal code form Corollary \ref{cor:[m(m+1)/2,m]} would be
better than that of the code given in (\ref{eq:P and P+P}), and suggests the following conjecture. 

\begin{figure}
\centering
\begin{tikzpicture}
	\begin{axis}[
	    xmin=5,
	    xtick={6,8,...,16},
	    minor x tick num=1,
	    ymin=5,
	    ytick={10,20,...,50},
	    minor y tick num=1,
	    xmax=16,
	    ymax=50,
	    xlabel=Dimension $m$,
	    ylabel=Minimum Distance,
		height=8cm,
		width=8cm,
		grid=both,
         grid style={line width=.2pt, draw=gray!50},
          legend style={at={(0.2,0.95)},
                       anchor=north,legend columns=1},
	]
	\addplot[style={densely dashed},blue, mark=square*] coordinates {
		(5,5)
		(6,7)
		(7,11)
		(8,13)
		(9,16)
		(10,21)
		(11,25)
		(12,29)
		(13,34)
		(14,39)
		(15,44)
		(16,50)
	};
	\addlegendentry{\tiny{$d_{max}(m)$}}
	 		
	\addplot[red,mark=diamond*] coordinates{
	    (5,5)
		(6,6)
		(7,9)
		(8,11)
		(9,13)
		(10,14)
		(11,15)
		(12,23)
		(13,24)
		(14,27)
		(15,27)
		(16,35)
	};
	\addlegendentry{\tiny{$d_{min}(m)$}}

	\addplot[style={loosely dashed},green,mark=star] coordinates{
		(5,5)
		(6,6)
		(7,7)
		(8,8)
		(9,9)
		(10,10)
		(11,11)
		(12,12)
		(13,13)
		(14,14)
		(15,15)
		(16,16)
	};
	\addlegendentry{\tiny{$d(m)$}}
	\end{axis}
\end{tikzpicture}
\caption{Minimum Distance of Codes in Corollary \ref{cor:[m(m+1)/2,m]} and (\ref{eq:P and P+P})} \label{fig:minimum distance}
\end{figure}
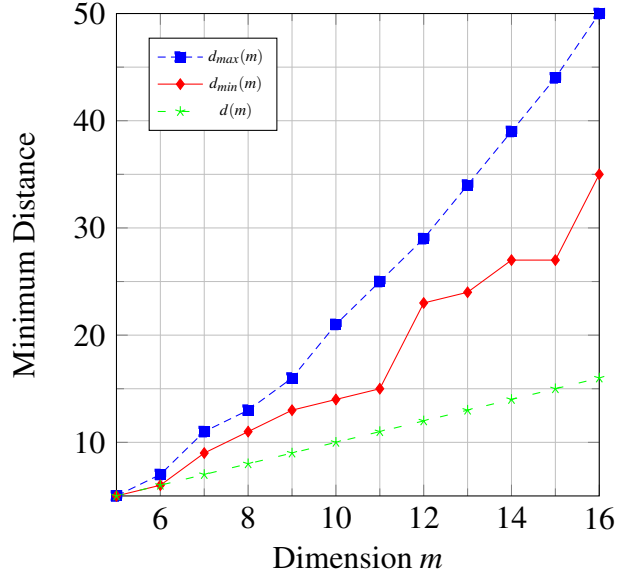

\begin{conj}
Let $m\ge 7$ be an integer. Then the minimal code in Corollary \ref{cor:[m(m+1)/2,m]}
has minimum distance greater than $m$. Moreover, if two primitive elements $\alpha$ and $\alpha'$ of $\gf(2^m)$ satisfy $\alpha' \neq \alpha^{2^i}$ for any $0 \le i \le m-1$, then the two codes corresponding to $\alpha$ and $\alpha'$  are inequivalent.
\end{conj}

By the definition of shortened codes, Proposition \ref{prop:prm-generator-alpha}
shows that the generator polynomial of the shortened second-order Reed-Muller code
$\srm (2,m)$ can be expressed as
\begin{eqnarray*}
 g_{2, \alpha}^*(X)= \prod_{ \scriptsize{\begin{array}{c}i_{m-1}+\cdots+ i_0 \le m-3
 \\ i_{m-1}, \cdots, i_0 \in \{0,1\} \end{array}}} (X-\alpha^{i_{m-1}2^{m-1}+ \cdots+i_0 2^0}).
 \end{eqnarray*}
For a positive integer $\epsilon$, let $\mathfrak C_{\epsilon}$ be the code contained in the shortened second-order Reed-Muller code $\srm(2,m)$   given by
\begin{eqnarray}\label{eq:modpoly-linear}
\mathfrak C_{\epsilon} =\left \{ \left ( \sum_{i=0}^{\epsilon-1}  c_i X^i \right ) g_{2, \alpha}^*(X)  : c_i \in \gf(2) \right \}.
\end{eqnarray}

The following  theorem provides a way of constructing linear codes of dimension $m$ and length less than $m(m+1)/2$ via sets of
algebraic immunity $2$.
\begin{theorem}\label{thm:C-[0;m(m+1)/2]}
Let $\epsilon$ be an integer with $1\le \epsilon < m(m-1)/2$.
 Let $\C$ be the set given by
\begin{eqnarray*}
\C=\left \{ \left ( \tr^m_1 \left (a \alpha^{0} \right ), \cdots, \tr^m_1 \left (a \alpha^{m(m+1)/2-\epsilon-1} \right )  \right ): a \in \gf(2^m) \right \}.
\end{eqnarray*}
Then $\C$ is a   minimal code
if and only if the minimum distance of the code $\mathfrak C_{\epsilon}$ in (\ref{eq:modpoly-linear}) is greater than $2^{m-2}$.

 \end{theorem}
\begin{proof}
Since $1\le \epsilon < m(m-1)/2$, it follows immediately that $\mathrm{dim}(\C)=m$.

Let us first prove the sufficient condition for $\C$ to be a minimal code.
Suppose the assertion of the theorem is false.
There would be
a codeword $\left ( \sum_{i=0}^{\epsilon -1} c_iX^i \right ) g^*_{2,\alpha}(X) $  of weight $2^{m-2}$ of $ \mathfrak C_{\epsilon}$
, where $c_i \in \gf(2)$. Set $g(X)=X^{m(m+1)/2- \epsilon}\left ( \sum_{i=0}^{\epsilon -1} c_iX^i \right ) g^*_{2,\alpha}(X)$
and let $f\in \mathbb B_{m}^0$ be the Boolean function corresponding to $g(X)$.
Then $g(X)$ is a codeword of weight $2^{m-2}$ of  $\srm(2,m)$ and $f$ satisfies
\begin{eqnarray}\label{eq:f-annihilator}
f(\alpha^i)=0,
\end{eqnarray}
where $0 \le i \le m(m+1)/2- \epsilon-1$.
Lemma \ref{lem:minimum-set-geometry} now leads to $f(x)=\tr^m_1(a_1x) \tr^m_1(a_2x)$, where $a_1\neq a_2 \in \gf(2^m)^*$.
Let $\mathbf{c}_i$ be the nonzero codeword of $\C$ given by $\left ( \tr^m_1(a_i \alpha^{i}) \right )_{i=0}^{m(m+1)/2- \epsilon-1}$, where $i=1 \text{ or } 2$.
From (\ref{eq:f-annihilator}), it is a simple matter to check that $\supp (\mathbf{c}_1+ \mathbf{c}_2)= \supp (\mathbf{c}_1) \dot \cup \supp ( \mathbf{c}_2) $,
which contradicts the assumption that $\C$ is a minimal code. Consequently, the minimum distance of $\mathfrak C_{\epsilon}$
is greater than $2^{m-2}$.

Conversely, let $\mathfrak C_{\epsilon}$ be a linear code of minimum distance greater than $2^{m-2}$.
Suppose that $\C$ is not a minimal code. Then we could find two distinct nonzero codewords $\mathbf{c_0}$ and $\mathbf c_1$ of $\C$
such that $\supp ( \mathbf{c}_1) \subsetneq \supp (\mathbf{c}_0)$, where $\mathbf c_i= \left ( \tr^m_1(a_i \alpha^i) \right )_{i=0}^{m(m+1)/2-\epsilon-1}$.
Let $f$ be the quadratic Boolean function $\tr^m_1( a_1x) \tr^m_1(a_2 x)$, where $a_2=a_0+a_1$.
A trivial verification shows that
\begin{eqnarray}\label{eq:f-0-sequen=0}
f(\alpha^0)=\cdots =f(\alpha^{m(m+1)/2-\epsilon-1})=0,
\end{eqnarray}
and
\begin{eqnarray}\label{eq:wt(f)=q/4}
\wt (f) =2^{m-2}.
\end{eqnarray}
Since $g^*_{2,\alpha}(X)$ is the generator polynomial of $\srm(2,m)$, the codeword $\left ( f(P_i) \right )_{i=1}^{2^m-1}$
can be uniquely expressed as $\left (\sum_{i=0}^{m(m+1)/2-1} c_iX^i \right ) g^*_{2,\alpha}(X)$, where $c_i\in \gf(2)$.
From (\ref{eq:f-0-sequen=0}) it may be concluded that $c_i=0$ for any $0 \le i \le m(m+1)/2-\epsilon-1$.
Combining this with (\ref{eq:wt(f)=q/4}) we deduce that
 $$\wt \left (\left (\sum_{i=0}^{\epsilon-1} c_{i+m(m+1)/2-\epsilon}X^i \right ) g^*_{2,\alpha}(X)  \right )=2^{m-2}.$$
This contradicts our assumption about the minimum distance of $\mathfrak C_{\epsilon}$.
It completes the proof.

\end{proof}

Let $\alpha$ be any primitive element of $\gf(2^m)$. Let us denote by $\epsilon_m(\alpha)$
the maximum $\epsilon$ such that the code $\C$ in Theorem \ref{thm:C-[0;m(m+1)/2]} is an $m$-dimensional minimal code.
The values of $\epsilon_m(\alpha)$ and their  corresponding frequencies  are listed in
Table \ref{tab:defect}  for $5 \le m \le 9$. It shows that a large number of minimal codes with dimension $m$ and length less than $m(m+1)/2$
can be produced from Theorem \ref{thm:C-[0;m(m+1)/2]}.

\begin {table}
\begin{center}
\begin{tabular}{>{\sf }lcc|lcc}    %
\toprule
 $m$ & $\epsilon_{m}(\alpha)$ & Freq. & \cellcolor{mygray} $m$ & \cellcolor{mygray} $\epsilon_{m}(\alpha)$ &  \cellcolor{mygray} Freq. \\
\rowcolor{mygray}
5     & 0  & 10 &   \cellcolor{white} 9 &  \cellcolor{white} 1 & \cellcolor{white} 36 \\
5       & \cellcolor{red} 1  & 20& \cellcolor{mygray} 9 & \cellcolor{mygray} 5 & \cellcolor{mygray} 18   \\
\rowcolor{mygray}
6    & 0  & 12 & \cellcolor{white} 9 & \cellcolor{white}6 &\cellcolor{white} 18\\
6 & 1 & 12 & \cellcolor{mygray} 9 & \cellcolor{mygray} 7 & \cellcolor{mygray} 18 \\
\rowcolor{mygray}
6    & \cellcolor{red}3  & 12& \cellcolor{white}  9 & \cellcolor{white} 8 & \cellcolor{white} 36 \\
7 & 0 & 42 &\cellcolor{mygray}  9 &\cellcolor{mygray}  9 &\cellcolor{mygray}  36\\
\rowcolor{mygray}
7 & 3 & 28  & \cellcolor{white} 9 & \cellcolor{white} 10 & \cellcolor{white} 72 \\
7     & 4  & 42 &  \cellcolor{mygray} 9& \cellcolor{mygray} 11 & \cellcolor{mygray} 36\\
\cellcolor{mygray}7     & \cellcolor{red}5  & \cellcolor{mygray}14 &  ~9 &  12 &54 \\
\rowcolor{mygray}
\cellcolor{white}8    & \cellcolor{white} 6  & \cellcolor{white} 48 &  ~9 &  13 & 18 \\
\cellcolor{mygray}8 & \cellcolor{mygray}7 & \cellcolor{mygray} 16 &  ~9 &  14 & 54 \\
\rowcolor[gray]{0.9}
\cellcolor{white}8    & \cellcolor{red}10  & \cellcolor{white} 64 &  ~9 & \cellcolor{red} 15 &18\\
\rowcolor[gray]{0.9}
9 & 0 & 18 \\
\bottomrule
\end{tabular}
\caption {Value Distribution of $\epsilon_{m}(\alpha)$ ($5\le m \le 9$)} \label{tab:defect}
\end{center}
\end {table}


As a corollary of Theorem \ref{thm:C-[0;m(m+1)/2]}, we have the following.

\begin{corollary}\label{cor:C-[0;m(m+1)/2-2]}
Let $\C$ be the set given by
\begin{eqnarray*}
\C=\left \{ \left ( \tr^m_1 \left (a \alpha^{0} \right ), \cdots, \tr^m_1 \left (a \alpha^{m(m+1)/2-2} \right )  \right ): a \in \gf(2^m) \right \}.
\end{eqnarray*}
Then $\C$ is    minimal
if and only if the Hamming weight of the generator polynomial $g_{2, \alpha}^*(X)$
of the shortened second-order Reed-Muller code $\srm (2, m)$ is not equal to $2^{m-2}$.
 \end{corollary}

 \begin{corollary}\label{cor:C-[0;m(m+1)/2-3]}
Let $\C$ be the set given by
\begin{eqnarray*}
\C=\left \{ \left ( \tr^m_1 \left (a \alpha^{0} \right ), \cdots, \tr^m_1 \left (a \alpha^{m(m+1)/2-3} \right )  \right ): a \in \gf(2^m) \right \}.
\end{eqnarray*}
Then $\C$ is a minimal  code if and only if  both $\wt \left( g_{2, \alpha}^*(X) \right )$ and
$\wt \left( (1+X) g_{2, \alpha}^*(X) \right )$ are greater than $2^{m-2}$.
 \end{corollary}

\begin{example}
 Let $q=2^5$ and $\alpha$ be a primitive element with minimal polynomial
 $ \alpha^5 + \alpha^3 + 1=0$. Then $g^*_{2,\alpha}(X)=X^{16} + X^{12} + X^{11} + X^{10} + X^9 + X^4 + X + 1$. Clearly, $\wt \left ( g^*_{2,\alpha}(X) \right )=8$.
    The binary linear code $\C$ in Corollary \ref{cor:[m(m+1)/2,m]} is a minimal code, but that in Corollary \ref{cor:C-[0;m(m+1)/2-2]} is not.
 \end{example}

 \begin{example}
 Let $q=2^6$ and $\alpha$ be a primitive element with minimal polynomial
 $\alpha^6 + \alpha^5 + \alpha^3 + \alpha^2 + 1=0$. Then $g^*_{2,\alpha}(X)=X^{42} + X^{41} + X^{39} + X^{38} + X^{37} + X^{32} + X^{31} + X^{30} + X^{29} + X^{24} + X^{19} +
    X^{17} + X^{16} + X^{13} + X^{12} + X^{11} + X^{10} + X^9 + X^8 + X^7 + X^3 + X^2 + X +
    1$ and $(1+X)g^*_{2,\alpha}(X)=X^{43} + X^{41} + X^{40} + X^{37} + X^{33} + X^{29} + X^{25} + X^{24} + X^{20} + X^{19} + X^{18} +
    X^{16} + X^{14} + X^7 + X^4 + 1$. Clearly, $\wt \left ( g^*_{2,\alpha}(X) \right )=24$ and $\wt \left ( (1+X) g^*_{2,\alpha}(X) \right )=16$.
    The binary linear code $\C$ in Corollary \ref{cor:C-[0;m(m+1)/2-2]} is a minimal code, but that in Corollary \ref{cor:C-[0;m(m+1)/2-3]} is not.
 \end{example}

\begin{example}
 Let $q=2^6$ and $\alpha$ be a primitive element with minimal polynomial
 $\alpha^6 + \alpha^5 + \alpha^4 + \alpha + 1=0$. Then $g^*_{2,\alpha}(X)=X^{42} + X^{41} + X^{39} + X^{37} + X^{36} + X^{35} + X^{33} + X^{30} + X^{27} + X^{26} + X^{24} +
    X^{21} + X^{19} + X^{17} + X^{16} + X^{15} + X^{13} + X^{9} + X^{8} + X^7 + X^5 + X^4 + X^3
    + 1$ and $(1+X)g^*_{2,\alpha}(X)=X^{43} + X^{41} + X^{40} + X^{39} + X^{38} + X^{35} + X^{34} + X^{33} + X^{31} + X^{30} + X^{28} +
    X^{26} + X^{25} + X^{24} + X^{22} + X^{21} + X^{20} + X^{19} + X^{18 }+ X^{15} + X^{14} + X^{13} +
    X^{10} + X^7 + X^6 + X^3 + X + 1$. Clearly, $\wt \left ( g^*_{2,\alpha}(X) \right )=24$ and $\wt \left ( (1+X) g^*_{2,\alpha}(X) \right )=28$.
    Both the binary linear codes in Corollary \ref{cor:C-[0;m(m+1)/2-2]} and \ref{cor:C-[0;m(m+1)/2-3]} are  minimal codes.
 \end{example}

It would be interesting to know how the Hamming weight  of the polynomial $ g_{2, \alpha}^*(X)$
 would be affected by selecting $\alpha$.
Based on our numerical experiments, we pose the following conjecture and open problem.
\begin{conj}
For any integer  $m\ge 5$ ,  there exists a primitive element $\alpha$ of $\gf(2^m)$
such that the Hamming weight of the generator polynomial $ g_{2, \alpha}^*(X)$
of the shortened second-order Reed-Muller $\srm (2,m)$ is greater than $2^{m-2}$.
\end{conj}

\begin{open}
Are there infinitely many positive integers $m$ such that $\wt \left ( g_{2, \alpha}^*(X) \right ) > 2^{m-2}$
for any primitive element $\alpha$ of $\gf(2^m)$?
\end{open}

Let $\left ( f(P_i) \right )_{i=0}^{q-1}$ be any codeword of the second-order Reed-Muller code $\RM(2,m)$. Then
 the corresponding Boolean function $f$ can be uniquely expressed as
\begin{eqnarray}\label{eq:rm2-trance}
f(x)=\left \{
\begin{array}{cl}
\begin{array}{l}\tr^{m/2}_1\left (a_{m/2} x^{2^{m/2}+1} \right )+ \\
 \sum \limits_{i=1}^{ \frac{m-2}{2} } \tr^m_1\left (a_i x^{2^i+1} \right )+\tr^m_1(a_0 x)  +c\end{array}, & \text{ if $m$ is even},  \\
& \\
\sum \limits_{i=1}^{ \frac{m-1}{2}} \tr^m_1\left (a_i x^{2^i+1} \right )+\tr^m_1(a_0 x)  +c, & \text{ if $m$ is odd},
\end{array}
\right .
\end{eqnarray}
where $a_{m/2} \in \gf(2^{m/2})$, $c \in \gf(2)$ and $a_i\in \gf(2^m)$ for $0\le i \le \lfloor (m-1)/ 2 \rfloor$.
Berlekamp and Sloane \cite{Berlekamp69} have shown that all possible weights of codewords of $\RM(2,m)$ are of the forms
$2^{m-1}$ and $2^{m-1}\pm 2^{m-1-j}$, where $0\le j \le \lfloor m /2 \rfloor$. A compact formula
of the weight distribution of $\RM(2,m)$ can be found in \cite{Li19}.

Next  we shall present several   classes of minimal codes contained in the
shortened  second-order Reed-Muller code $\srm (2,m)$.

\begin{theorem}\label{thm:quadratic functions space-odd}
Let $m\ge 3$ be an odd integer. Let $\C$ be the cyclic subcode of the shortened  second order Reed-Muller code $\srm(2,m)$ given by
\begin{eqnarray*}
\left \{ \left ( \sum_{i=1}^{( m-1)/2 } \tr^{m}_1 \left (a_i \alpha^{(2^i+1)j} \right )  \right )_{j=0 }^{ 2^m-2}: a_i \in \gf(2^m) \right \}.
\end{eqnarray*}
Then $\C$ is a minimal code with parameters $\left [ 2^m-1, m(m-1)/2, \ge  3\cdot 2^{m-3} \right ]$.
\end{theorem}

\begin{proof}
Let us first prove that there is no codeword of weight $2^{m-2}$ or $3\cdot 2^{m-2}$.
If there existed a codeword $\mathbf{c} \in \C$ such that $\wt(\mathbf{c})=2^{m-2}$,
by Lemma \ref{lem:minimum-set-geometry} there would be two distinct elements $a, b \in \gf(2^m)^*$
such that $\mathbf{c}= \left ( \tr^m_1(a \alpha^j) \tr^m_1(b \alpha^j) \right )_{j=0}^{2^m-2}$.
There is no loss of generality in assuming $a=1$ and $b\in \gf(2^m)\setminus \gf(2)$ as $\C$ is a cyclic code.
A direct calculation  shows
\begin{eqnarray*}
\lefteqn{ \tr^m_1( \alpha^j) \tr^m_1 (b \alpha^j) } \\
&=& \tr^m_1 \left(b \alpha^j \tr^m_1( \alpha^j) \right )\\
&=&  \sum^{m-1}_{i=0} \tr^m_1 \left (b \alpha^{(2^i+1)j} \right )\\
&=& \sum^{m-1}_{i=1} \tr^m_1 \left (b \alpha^{(2^i+1)j} \right )+ \tr^m_1 \left (\sqrt{b} \alpha^{j} \right ), 
\end{eqnarray*}
which is impossible. Thus there is no codeword of weight $2^{m-2}$. Suppose there was a codeword $\mathbf{c} \in \C$ of
weight $3\cdot 2^{m-2}$, then the codeword $\mathbf{1}+ \mathbf{c}$ of $\RM(2,m)$ has weight $2^{m-2}$.
Proposition \ref{prop:mini-geometry}  now implies $\mathbf{1}+ \mathbf{c}= \left ((1+ \tr^m_1(a \alpha^j))(1+ \tr^m_1(b \alpha^j)) \right )_{j=0}^{2^m-2}$,
where $a\neq  b \in \gf(2^m)^*$. We can assume that $a=1$ and $b \not \in \gf(2)$ as in the previous discussion.
Consequently,  $\mathbf{c}$ is just the codeword given by the Boolean function $\tr^m_1( x)\tr^m_1(b x)+\tr^m_1((b+1) x)$,
which can be rewritten as:
\begin{eqnarray*}
\begin{array}{rl}
& \sum^{m-1}_{i=1} \tr^m_1 \left (bx^{2^i+1} \right )+ \tr^m_1 \left ( (b+\sqrt{b}+1) x \right ).
\end{array}
\end{eqnarray*}
Since $m$ is an odd integer, we have $b+\sqrt{b}+1\neq 0$. This clearly forces $\mathbf{c} \not \in \C$, a contradiction.
Hence the weight  $\wt(\mathbf c) $ of any codeword $\mathbf{c}$ of $\C$  is not equal to $3\cdot 2^{m-2}$.
Consequently, we conclude that for any nonzero codeword $\mathbf{c}$ of $\C $ its weight satisfies the following
\begin{eqnarray*}
3\cdot 2^{m-3} \le \wt (\mathbf c ) \le 5\cdot 2^{m-3}.
\end{eqnarray*}
Therefore $\C$ is a minimal code by Lemma \ref{lem:AB} and has minimum distance at least $3\cdot 2^{m-3}$.
It is obvious that $\mathrm{dim}(\C) =m(m-1)/2$ from the definition of $\C$. This
completes the proof.
\end{proof}

The proof above gives more, namely the code $\C$ in Theorem \ref{thm:quadratic functions space-odd}
is not a minimal code if $m\ge 4$ is an even integer.

\begin{theorem}\label{thm:quad+linear}
Let $m\ge 3$ be an  integer. Let $\C$ be the cyclic subcode of the shortened  second-order Reed-Muller code $\srm(2,m)$ given by
\begin{eqnarray*}
\left \{ \left ( \sum_{i=2}^{\lfloor m /2 \rfloor } \tr^{m}_1 \left (a_i \alpha^{(2^i+1)j}\right ) + \tr^m_1(b \alpha^j)  \right )_{j=0}^{2^{m}-2}: b, a_i \in \gf(2^m) \right \}.
\end{eqnarray*}
Then $\C$ is a minimal code with parameters $\left [ 2^m-1, m(m-1)/2, \ge 3\cdot 2^{m-3} \right ]$.
\end{theorem}
 \begin{proof}
 The statements will be proved once we prove that there are no codewords of weight $2^{m-2}$ or $3\cdot 2^{m-2}$.
 Suppose, contrary to our claim, that there exists a codeword with weight $2^{m-2}$ or $3\cdot 2^{m-2}$.
 By a similar argument in the proof of Theorem \ref{thm:quadratic functions space-odd},  we could find
 $b\in \gf(2^m) \setminus \gf(2)$ and $c \in \gf(2)$ such that
 the codeword of $\srm(2,m)$ corresponding to the Boolean function $f(x)=\tr^m_1( x)\tr^m_1(b x)+c\tr^m_1((b+1) x)$
 lies in $\C$. A simple calculation yields
 \begin{eqnarray*}
\lefteqn{ \tr^m_1( x)\tr^m_1(b x)+c\tr^m_1((b+1) x) } \\
&=& \sum_{i=0}^{m-1} \tr^m_1(b x^{2^i+1}) +c\tr^m_1((b+1) x)\\
&=& \sum_{i=2}^{m-2} \tr^m_1(b x^{2^i+1}) +\tr^m_1((\sqrt{b}+bc+c) x)\\
&  &+ \tr^{m}_1((b+b^2) x^{2+1}),
\end{eqnarray*}
which contradicts the definition of $\C$.
This completes the proof.
 \end{proof}

In the spirit of Theorems   \ref{thm:quadratic functions space-odd} and
\ref{thm:quad+linear}, we pose the following open problem.
\begin{open}
Does there exist a minimal code $\C$ contained in the  second-order Reed-Muller code $\RM (2,m)$
such that its dimension $\mathrm{dim}(\C)$ is greater than $m(m-1)/2$?
\end{open}

The following two theorems describe two infinite classes of minimal codes obtained by puncturing of
the minimal codes in Theorems \ref{thm:quadratic functions space-odd} and \ref{thm:quad+linear}.

\begin{theorem}
Let $m\ge 5$ be an odd integer. Let $\delta $ be an integer with $\sum_{i=1}^{t-1} \binom{m}{i}  \le \delta  < \sum_{i=1}^{t} \binom{m}{i}$  and $5\le t \le m$.
Let $\C$ be the binary code given by
\begin{eqnarray*}
\left \{ \left ( \sum_{i=1}^{ (m-1)/2 } \tr^{m}_1 \left (a_i \alpha^{(2^i+1)j} \right )  \right )_{j=0  }^{ \delta-1}: a_i \in \gf(2^m) \right \}.
\end{eqnarray*}
Then $\C$ is a minimal code with parameters $\left [ \delta , m(m-1)/2, \ge \sum_{i=0}^{t-3} \binom{m-2}{i} \right ]$.
\end{theorem}
\begin{proof}
Combining  Theorem \ref{thm:quadratic functions space-odd} and Lemma \ref{lem:[h;delta]+0} with Corollary \ref{cor:ad>2l} proves  the desired conclusion.
\end{proof}

\begin{theorem}
Let $m\ge 5$ be an integer. Let $\delta $ be an integer with $\sum_{i=1}^{t-1} \binom{m}{i}  \le \delta  < \sum_{i=1}^{t} \binom{m}{i}$  and $5\le t \le m$.
Let $\C$ be the binary code given by
\begin{eqnarray*}
\left \{ \left ( \sum_{i=2}^{ \lfloor m/2 \rfloor } \tr^{m}_1 \left (a_i \alpha^{(2^i+1)j} \right )+ \tr^m_1 \left (a_0 \alpha^j \right )   \right )_{j=0  }^{\delta-1}: a_i \in \gf(2^m) \right \}.
\end{eqnarray*}
Then $\C$ is a minimal code with parameters $\left [ \delta , m(m-1)/2, \ge \sum_{i=0}^{t-3} \binom{m-2}{i} \right ]$.
\end{theorem}

\begin{proof}
Combining  Theorem \ref{thm:quad+linear} and Lemma \ref{lem:[h;delta]+0} with Corollary \ref{cor:ad>2l} yields the desired conclusion.
\end{proof}

\section{Minimal codes from vectorial Boolean functions with high algebraic immunity}\label{sec:codes-vec-func}

In this section, we shall demonstrate that binary minimal codes can be obtained from the vector subspace spanned by
certain subcodes of Reed-Muller
and the component functions of vectorial Boolean functions with high algebraic immunity.

For a vectorial Boolean $(m,r)$-function $F=(f_1, \cdots, f_r)$  with $\ai (F)=t\ge 1$, let $\mathrm{Span}(F)$ be the linear code defined by
\begin{eqnarray}
\mathrm{Span}(F)=\left \{ \left (\sum_{j=1}^{r} a_j f_j(P_i) \right)_{i=0}^{2^m-1}: a_j \in \gf (2)  \right \}.
\end{eqnarray}
Let $\C$ be a linear code. The sum of two linear subcodes $\C_1$ and $\C_2$ of $\C$ is the set, denoted $\C_1 + \C_2$, consisting of all the elements
$\bc_1+ \bc_2$, where $\bc_1 \in \C_1$ and $\bc_2 \in \C_2$. If $\C_1 \cap  \C_2 =\{\mathbf{0}\}$, then
the sum  is also called the direct sum of $\C_1$ and $\C_2$, and is written by $\C_1 \bigoplus \C_2 $.
 Note that the direct sum of linear subcodes of a linear code is not the same thing as
 the direct sum of some linear codes.
\begin{lemma}\label{lem:vec-f-set}
 Let $r\ge 2$ and $F$ be a vectorial Boolean  $(m,r)$-function with $\ai (F)=t\ge 1$. Let $D$ be the subset of $\gf(2^m)$ given by
 $$D=\left \{ x\in \gf(2^m): v_1\cdot F(x)= \epsilon_1, v_2\cdot F(x)= \epsilon_2   \right \},$$
 where $v_1$ and $ v_2 $ are two distinct nonzero  elements in $\gf(2)^r$ and $\epsilon_1, \epsilon_2 \in \gf(2)$.
 Then $\ai(D) \ge t$.
 \end{lemma}

\begin{proof}
By assumption, $v_1$ and $v_2$ are linearly independent over $\gf(2)$.
It follows that there exists $y\in \gf(2)^r$ such that
$v_1\cdot y =\epsilon_1$ and $v_2\cdot y =\epsilon_2$.
We thus get $F^{-1}(y) \subseteq D$.
By the definition of algebraic immunity, $\ai(D) \ge t$.
This completes the proof.
\end{proof}

Let $F$ be a vectorial Boolean $(m,r)$-function with $\ai(F)\ge t$. Then $\mathrm{Span}(F) \cap \RM(t-1,m)=\{\mathbf 0\}$ from Lemma \ref{lem:vec-f-set}.
The following theorem presents a new method to  construct minimal codes from some subcodes of Reed-Muller codes and vectorial Boolean functions with
high algebraic immunity.

 \begin{theorem}\label{thm:min-rm-Vecfun}
 Let $\C$ be a $k$-dimensional subcode of the Reed-Muller code $\rm{RM} (\ell, m)$ such that $k>1$. Let $F$ be a vectorial Boolean $(m,r)$-function
 with algebraic immunity $\ai(F) \ge 2\ell+1$. Then $\C \bigoplus \mathrm{Span}(F)$ is a  minimal code of dimension $k+r$.
  \end{theorem}
  
  \begin{proof}
  It is clear that $\mathrm{dim}(\C \bigoplus \mathrm{Span}(F))=k+r$. It remains to prove that $\C \bigoplus \mathrm{Span}(F)$ is a
  minimal code.

  Let $\left ( v_1 \cdot F(P_i) +f_1(P_i) \right )_{i=0}^{2^m-1}$ and $\left ( v_2 \cdot F(P_i) +f_2(P_i) \right )_{i=0}^{2^m-1}$ be any two nonzero codewords
  of $\C \bigoplus \mathrm{Span}(F)$, where $f_1, f_2 \in \mathbb B_m$ are Boolean functions of algebraic degree at most $\ell$,
  and $v_1, v_2 \in \gf(2)^r$.
  We will complete the proof of the theorem if we prove the following:
  \begin{eqnarray}\label{eq:c*c-f*f}
  \left ( v_1 \cdot F(x) +f_1(x) \right ) \cdot \left ( v_2 \cdot F(x) +f_2(x) \right ) \not \equiv 0.
  \end{eqnarray}
To this end, consider the following four cases.

Case 1:  $\left ( v_1 \cdot F +f_1\right ) = \left ( v_2 \cdot F +f_2 \right ).$
If $f_1 \equiv 0$, then $v_1 \neq \mathbf{0}$. Applying Lemma \ref{lem:vec-f-set},
we see that $\supp (  v_1 \cdot F )$ is not the  empty set, which gives (\ref{eq:c*c-f*f}).
Let $f_1 $ be a nonzero function. By the assumption of the theorem,  $f_1$ does not vanish on
$\left \{ x \in \gf(2^m):  v_1 \cdot F(x)=0 \right \}$, and (\ref{eq:c*c-f*f}) is proved.

Case 2: $f_1=f_2 \equiv 0$. In this case, none of  $v_1$ and $ v_2$ is the zero vector.
Lemma \ref{lem:vec-f-set} now leads to
 $\supp (v_1 \cdot F) \cap \supp ( v_2 \cdot F) \neq \emptyset$.
Then it follows that the product of $v_1 \cdot F(x)$ and $v_2 \cdot F(x)$ is not the zero function, which establishes (\ref{eq:c*c-f*f}).

Case 3: $f_1 \equiv 0$ and  $f_2 \not  \equiv 0$, or , $f_1 \not \equiv 0$ and  $f_2  \equiv 0$.
 By symmetry, we can assume $f_1 \equiv 0$ and  $f_2 \not  \equiv 0$. Let $D$ be the subset of $\gf(2^m)$
 given by
 \begin{eqnarray*}
 D=\left \{
 \begin{array}{rl}
 \{x\in \gf(2^m): v_1 \cdot F=1, v_2 \cdot F=1 \}, &\text{ if } v_1 =v_2,\\
 & \\
 \{x\in \gf(2^m): v_1 \cdot F=1, v_2 \cdot F=0 \}, &\text{ if } v_1 \neq v_2.
 \end{array}
 \right .
 \end{eqnarray*}
Note that $f_1\not \equiv 1$ since $\mathrm{dim}(\C) >1$.
Therefore none of $f_1$ and $(1+f_1)$ vanishes on $D$ from Lemma \ref{lem:vec-f-set}. Thus
(\ref{eq:c*c-f*f}) holds.

Case 4: $f_1 \not \equiv 0$ and $f_2 \not \equiv 0$. Denote $D=\{x\in \gf(2^m): v_1 \cdot F=v_2 \cdot F=0\}$.
It is obvious that $f_1f_2 \not \equiv 0$ because $\C$ is a minimal code.
Combining Lemma \ref{lem:vec-f-set} with $\mathrm{deg}(f_1f_2)\le 2\ell$ yields $f_1f_2 \not \in \mathrm{Ann} (D)$, which gives (\ref{eq:c*c-f*f}).

Summarising the discussions in the four cases completes the proof of (\ref{eq:c*c-f*f}) and the theorem. 
\end{proof}

In order to apply Theorem \ref{thm:min-rm-Vecfun} to obtain minimal codes, we need to construct vectorial Boolean functions 
with hight algebraic immunity. The following result is in this direction. 
 
 \begin{theorem}\label{thm:vecFUN-AI}
Let $n_0$, $n_1$, $\cdots$, $n_{2^r}$ be integers satisfying $0=n_0 < n_1 < \cdots <n_{2^r}=2^m-1$.
Let $y_0, y_1, \cdots, y_{2^r-1}$ be an enumeration of the points of $\gf(2)^r$. Let $F$ be the function
defined by
\begin{eqnarray*}
F(x)=\left \{
\begin{array}{rl}
y_i, & \text{ if }x\in [n_i;n_{i+1}-n_{i}]_{\alpha},\\
\\
y_0, & \text{ if } x=0.\\
\end{array}
\right .
\end{eqnarray*}
Then $F$ is a vectorial Boolean $(m,r)$-function with algebraic immunity $t$,
where $t$  is the biggest  integer $t$ such that $\sum_{j=1}^{t-1} \binom{m}{j} \le n_{1}-n_0$ and
 $\sum_{j=0}^{t-1} \binom{m}{j} \le n_{i+1}-n_i$ for $1\le i \le 2^r-1$.
\end{theorem} 

\begin{proof}
The desired conclusion follows directly from Lemmas \ref{lem:ai-det-0} and \ref{lem:[h;delta]+0}.
\end{proof}

 Note that the theorem is still true if the vector space $\gf(2)^r$
 is replaced  by the finite field $\gf(2^r)$.
 Using Theorem \ref{thm:vecFUN-AI}, we obtain the following explicit construction  of vectorial functions 
 with high algebraic immunity.
 
\begin{corollary}\label{cor:vecFUN-AI}
Let $n_0$, $n_1$, $\cdots$, $n_{2^r}$ be integers satisfying $0=n_0 < n_1 < \cdots <n_{2^r}=2^m-1$.
Let $F$ be the vectorial Boolean function defined by $F(0)=0$ and $F(\alpha^j)=(y_0, \cdots, y_{r-1})$, where
\begin{eqnarray*}
n_{ \sum_{i=0}^{r-1} y_{i} 2^{i}} \le j < n_{1+ \sum_{i=0}^{r-1} y_{i} 2^{i}} .
\end{eqnarray*}
Then $F$ is a vectorial Boolean $(m,r)$-function with algebraic immunity $t$,
where $t$  is the biggest  integer $t$ such that $\sum_{j=1}^{t-1} \binom{m}{j} \le n_{1}-n_0$ and
 $\sum_{j=0}^{t-1} \binom{m}{j} \le n_{i+1}-n_i$ for $1\le i \le 2^r-1$.
\end{corollary}

\begin{corollary}\label{cor:simplex+vecfun}
Let $F$ be the vectorial Boolean $(m,r)$-function of Theorem \ref{thm:vecFUN-AI}
with $m^2+m+2\le 2^{m-r+1}$ and $n_1-n_0+1=n_2-n_1=\cdots =n_{2^r}-n_{2^r-1}=2^{m-r}$.
Let $\C(F)$ be the binary code given by
\begin{eqnarray*}
\C(F)=
 \left \{ \left ( v\cdot F(P_i) + \tr^m_1 (b P_i) \right )_{i=1}^{2^m-1}:
  v \in \gf(2)^r, b \in \gf(2^m) \right \}.
\end{eqnarray*}
Then $\C(F)$ is a minimal code of dimension $m+r$.
\end{corollary}
\begin{proof}
Combining Theorem \ref{thm:vecFUN-AI}  with Theorem \ref{thm:min-rm-Vecfun} proves the desired conclusion.
\end{proof}

\begin{corollary}
Let $\delta $ be an integer with $\sum_{i=0}^{t-1} \binom{m}{i}  \le \delta  < \sum_{i=0}^{t} \binom{m}{i}$ and $3\le t \le m-3$.
Let $f$ be the  Boolean function with $\supp (f)= [0;\delta]_{\alpha}$.
Let $\C(f)$ be the binary code defined by 
\begin{eqnarray*}
\C(f)=
 \left \{ \left ( f(P_i) + \tr^m_1 (b P_i) \right )_{i=1}^{2^m-1}:
  b \in \gf(2^m) \right \}.
\end{eqnarray*}
Then $\C(f)$ is a minimal code of parameters $[2^m-1,m+1,\ge d]$ with $d$ being the smaller  of $\delta$ and
\begin{eqnarray*}
 \max  \left \{\sum_{i=0}^{t-2} \binom{m-1}{i} +\sum_{i=0}^{t'-2} \binom{m-1}{i},  2^{m-1}- 1 - \frac{\ln 2}{\pi} (m+1) \sqrt{2^m} \right \},
\end{eqnarray*}
where   $t'=m-t$ when $\delta \neq \sum_{i=0}^{t-1} \binom{m}{i}$ and $t'=m-t+1$ when
$\delta = \sum_{i=0}^{t-1} \binom{m}{i}$.
\end{corollary}

\begin{proof}
An easy computation shows that
$$\sum_{i=0}^{m-t-1} \binom{m}{i} <  2^m-\delta \le \sum_{i=0}^{m-t}   \binom{m}{i}. $$
Note that $\supp (f+1) =2^m-\delta$.
From Lemmas \ref{lem:ai-det-0} and \ref{lem:[h;delta]+0} we conclude that
$\ai (\supp(f))=t$ and $\ai (\supp(f))=t'$. Theorem \ref{thm:min-rm-Vecfun} now implies that
$\C(f)$ is a minimal code of dimension $m+1$. We are left with the task of determining the lower bound of $d$.
It is easy to see that for any $b\in \gf(2^m)^*$ the Hamming distance $\mathrm{dist}(f,\tr^m_1(bx))$ of $f$ and $\tr^m_1(bx)$
can be written as
\begin{eqnarray*}
\begin{array}{c}
\mathrm{dist}(f,\tr^m_1(bx))=\wt \left (f (1+\tr^m_1(bx)) \right )+ \wt \left ((1+f) \tr^m_1(bx) \right ).
\end{array}
\end{eqnarray*}
By Lemma \ref{lem:wt-bound}, we deduce that
\begin{eqnarray}\label{eq:t-t'}
\mathrm{dist}(f,\tr^m_1(bx)) \ge \sum_{i=0}^{t-2} \binom{m-1}{i} +\sum_{i=0}^{t'-2} \binom{m-1}{i}.
\end{eqnarray}

An easy computation yields that
\begin{eqnarray}\label{eq:dist-walsh-zhang}
\begin{array}{rl}
\mathrm{dist}(f,\tr^m_1(bx)) = & 2^{m-1}+ \sum_{x\in \mathrm{Supp} (f)} (-1)^{\tr^m_1 (bx)}\\
=& 2^{m-1} + \sum_{i=0}^{\delta -1} (-1)^{\tr^m_1 (b \alpha^i)}\\
\ge & 2^{m-1}- 1 - \frac{\sqrt{2^m}}{\pi}  \ln \left ( \frac{4(2^m-1)}{\pi} \right )\\
\ge & 2^{m-1}- 1 - \frac{\ln 2}{\pi} (m+1) \sqrt{2^m},
\end{array}
\end{eqnarray}
where the first inequality follows from (\ref{eq:F-indpendentOFdelta}).
Combining (\ref{eq:t-t'}) and (\ref{eq:dist-walsh-zhang})
yields the desired conclusion.
\end{proof}

\begin{corollary}
Let $m$ be an odd integer.
Let $F$ be the vectorial Boolean $(m,r)$-function of Theorem \ref{thm:vecFUN-AI}
with $\sum_{i=0}^{4} \binom{m}{j}\le 2^{m-r}$ and $n_1-n_0+1=n_2-n_1=\cdots =n_{2^r}-n_{2^r-1}=2^{m-r}$.
Let $\C(F)$ be the binary code given by
\begin{eqnarray*}
\C(F)=\left \{
\begin{array}{r}
 \left ( v\cdot F(P_i) + \sum_{j=1}^{(m-1)/2}\tr^m_1 \left  (b_j P_i^{2^j+1} \right ) \right )_{i=1}^{2^m-1}: \\
 \\
  v \in \gf(2)^r, b_j \in \gf(2^m)
\end{array}
 \right \}.
\end{eqnarray*}
Then $\C(F)$ is a minimal code of dimension $\frac{m(m+1)}{2}+r$.
\end{corollary}
\begin{proof}
Combining Theorems \ref{thm:quadratic functions space-odd} and \ref{thm:vecFUN-AI}  with Theorem \ref{thm:min-rm-Vecfun} proves the desired conclusion.
\end{proof}

\begin{corollary}
Let $m$ be a positive integer.
Let $F$ be the vectorial Boolean $(m,r)$-function of Theorem \ref{thm:vecFUN-AI}
with $\sum_{i=0}^{4} \binom{m}{j}\le 2^{m-r}$ and $n_1-n_0+1=n_2-n_1=\cdots =n_{2^r}-n_{2^r-1}=2^{m-r}$.
Let $\C(F)$ be the binary code given by
\begin{eqnarray*}
\C(F)=\left \{
\begin{array}{r}
 \left ( v\cdot F(P_i) + \sum_{j=2}^{\lfloor m/2 \rfloor}\tr^m_1 \left  (b_j P_i^{2^j+1} \right ) +\tr^m_1(c P_i) \right )_{i=1}^{2^m-1}: \\
 \\
  v \in \gf(2)^r, b_j, c \in \gf(2^m)
\end{array}
 \right \}.
\end{eqnarray*}
Then $\C(F)$ is a minimal code of dimension $\frac{m(m+1)}{2}+r$.
\end{corollary}

\begin{proof}
Combining Theorems \ref{thm:quad+linear} and \ref{thm:vecFUN-AI}  with Theorem \ref{thm:min-rm-Vecfun} yields the desired conclusion.
\end{proof}

\begin{example}
 Let $q=2^7$ and $\alpha$ be a primitive element with minimal polynomial
 $ \alpha^7 + \alpha + 1=0$. Let $f_1$ be the Boolean function with $\supp(f_1)= [63;64]_{\alpha}$
 and let $f_2$ be the function with $\supp(f_2)=[31;32]_{\alpha} \cup [95;32]_{\alpha} $.
 Let $F$ denote the vectorial Boolean function $(f_1, f_2)$.
 Then the algebraic immunity of $F$ equals $3$ and
 the binary linear code $\C (F)$ in Corollary \ref{cor:simplex+vecfun} is a minimal code with parameters $[127,9,52]$.
 \end{example}

\section{Summary and concluding remarks} \label{sec:conc}

In this paper, a link between minimal linear codes and subsets of finite fields without nonzero
low degree annihilators was established. This link allowed us to construct binary minimal codes
with special sets, Boolean functions or vectorial Boolean with high algebraic immunity.
A general construction of minimal binary linear codes from sets without nonzero low degree annihilators  was proposed.
Employing this general construction, minimal codes of dimension $m$ and length less than or equal to $m(m-1)/2$ were  obtained,
 and a lower bounder on the minimum distance of the proposed minimal codes was established.
A explicit construction of minimal codes using certain subcodes of Reed-Muller codes
and vectorial Boolean functions with algebraic immunity was also developed.
These results show that there are natural connections  among
binary minimal codes, sets without nonzero low degree annihilators and
Boolean functions with high algebraic immunity.

The results of this paper were presented in terms of univariate representations of functions and codes. 
The corresponding multivariate analogies can be easily worked out.
It would be interesting to generalize the results of this paper to the nonbinary cases.
It would be good if the open problems and conjectures proposed  in this paper could be settled.
The reader is cordially invited to join this adventure.

\section*{Acknowledgements}
C. Ding’s research was supported by the Hong Kong Research Grants Council, Proj. No. 16300919.
S. Mesnager was supported by the ANR CHIST-ERA project SECODE.
C. Tang was supported by National Natural Science Foundation of China (Grant No.
11871058) and China West Normal University (14E013, CXTD2014-4 and the Meritocracy Research
Funds).

\renewcommand{\refname}{References}

\end{document}